\documentclass{svproc}

\usepackage{url}

\usepackage{hyperref}
\usepackage{type1cm}
\usepackage{makeidx}
\usepackage{graphicx}
\usepackage{multicol}        
\usepackage[bottom]{footmisc}
\usepackage{newtxtext}
\usepackage[varvw]{newtxmath}
\usepackage{bm}
\usepackage{siunitx}
\usepackage{enumerate}
\usepackage{algorithmic}
\usepackage{algorithm}
\usepackage{listings}

\begin{document}

\mainmatter

\title{Accelerating true orbit pseudorandom number generation using Newton's method}

\titlerunning{Accelerating true orbit pseudorandom number generation using Newton's method}

\author{Asaki Saito\inst{1} \and Akihiro Yamaguchi\inst{2}}

\authorrunning{Asaki Saito and Akihiro Yamaguchi}

\tocauthor{Asaki Saito and Akihiro Yamaguchi}

\institute{Future University Hakodate, 116-2 Kamedanakano-cho, Hakodate, Hokkaido 041-8655, Japan,\\
\email{saito@fun.ac.jp}
\and
Fukuoka Institute of Technology, 3-30-1 Wajiro-higashi, Higashi-ku, Fukuoka 811-0295, Japan,\\
\email{aki@fit.ac.jp}}

\maketitle

\begin{abstract}
The binary expansions of irrational algebraic numbers can serve as
high-quality pseudorandom binary sequences.
This study presents an efficient method for computing the exact binary
expansions of real quadratic algebraic integers using Newton's method.
To this end, we clarify conditions under which the first $N$ bits
of the binary expansion of an irrational number match those of its
upper rational approximation.
Furthermore, we establish that the worst-case time complexity of
generating a sequence of length $N$ with the proposed method is
equivalent to the complexity of multiplying two $N$-bit integers,
showing its efficiency compared to a previously proposed true orbit
generator.
We report the results of numerical experiments on computation time and
memory usage, highlighting in particular that the proposed method
successfully accelerates true orbit pseudorandom number generation.
We also confirm that a generated pseudorandom sequence
successfully passes all the statistical tests included in RabbitFile
of TestU01.
\end{abstract}

\section{Introduction}\label{sec:Introduction}

A pseudorandom number generator having good statistical properties
despite having a low computational cost is not only useful for various
applications such as simulation, numerical analysis, and secure
communications (see, e.g., \cite[Chapter 3]{Knuth}).
It is also interesting as an object of theoretical study since
the lower the memory usage and computation time, the
more difficult it usually becomes to generate pseudorandom sequences
having good statistical properties.
It does not seem true, however, that one only has to aim to develop a
generator having as low computational cost as possible that still
passes a certain set of standard statistical tests.
For example, a generator for evaluating an empirical statistical test
needs to have a higher level of statistical quality than that is
needed to merely pass some standard tests, as we shall describe below.

For empirical testing of a pseudorandom number generator, statistical
testing packages which consist of several standard tests are widely
used.
Some of well-known packages are DIEHARD \cite{Marsaglia}, NIST
statistical test suite \cite{NIST}, and TestU01 \cite{L'Ecuyer}.
Due to implementational or theoretical errors, however, several tests
in some packages have been reported to have defects (see, e.g.,
\cite{Brown} concerning DIEHARD and \cite{Hamano,HamanoKaneko,Okutomi}
concerning NIST statistical test suite).
This implies that it is desirable to verify the correctness of tests
in a package preferably beforehand.
However, if a pseudorandom number generator of marginal quality is
used for empirical verification of a test, then it is difficult to
distinguish whether the test or the generator is defective, even if
many failures are obtained by applying the test to the generator.
In order to avoid such a confusion, one needs to use a high-quality
pseudorandom number generator for the verification of a test.

We have devised pseudorandom number generators ---{\bf true orbit generators}--- using true orbits of
the Bernoulli map on irrational algebraic integers \cite{SaitoChaos2016,SaitoChaos2018}, in order to
generate pseudorandom sequences having as good statistical properties
as possible even if such generation increases the computational cost
to some extent.
Due to a countably infinite number of their possible states, the
true orbit generators can produce
nonperiodic sequences, which is in contrast with standard
generators:
Since usual generators have a finite number of possible states, they
can generate only (eventually) periodic sequences.
Other than the nonperiodicity, there are several mathematical supports
for the high statistical quality of the true orbit generators:
According to ergodic theory, the Bernoulli map can generate ideal
random binary sequences \cite{Billingsley}.
Also, Borel's conjecture \cite{Borel}, stating that every irrational
algebraic number is normal, is widely believed to be true in the field
of number theory.
Moreover, for an integer $b \ge 2$, the base-$b$ expansion of any
irrational algebraic number cannot have a regularity so simple that it
can be generated by a finite automaton \cite{Adamczewski} or by a
deterministic pushdown automaton \cite{Adamczewski2}.
However, the computational cost of generating pseudorandom sequences
using the true orbit generators is very
high.
Indeed, the time complexity of these generators is $O(N^2)$,
where $N$ is the length of a pseudorandom sequence to be generated, as
we observe in Appendix \ref{sec:AppendixComplexityTrueOrbitGenerators}.
Therefore, it is impractical to use these generators to generate a
long pseudorandom sequence on low-performance computers, such as those
having a slow CPU or limited memory.

In order to overcome this difficulty, in this paper, we employ
Newton's method, a technique for producing successively better
approximations to the roots of a function, to accelerate the true
orbit generator of \cite{SaitoChaos2016}.
This involves obtaining the exact binary expansion of a true root
(i.e., an algebraic integer of degree $2$) $\alpha$ from its
approximation $x$, which includes an error.
We establish a sufficient condition ensuring that the first $N$ bits
of the binary expansions of $\alpha$ and $x$ match, thereby ensuring
the generation of the same pseudorandom sequence as the true orbit
generator.
Furthermore, we demonstrate that the worst-case time complexity for
generating a sequence of length $N$ using the method proposed in this
study is equivalent to that of multiplying two $N$-bit integers,
showing its efficiency compared to the original generator with
$O(N^2)$ time complexity.
We also confirm, through numerical experiments, that the
acceleration of true orbit pseudorandom number generation has indeed
been achieved, and that a generated pseudorandom sequence of length
$2^{36} - 2$ successfully passes all the statistical tests included in
RabbitFile of TestU01.

\section{Preliminaries}\label{sec:Preliminaries}

In this section, we provide some definitions and results from
references \cite{SaitoChaos2016,SaitoPreprint}.

A complex number is called an {\bf algebraic integer} if it is a root of a
monic polynomial with (rational) integer coefficients (see, e.g.,
\cite[Chapter 5]{Hecke} for an explanation of algebraic integers).
If this polynomial is irreducible, then the degree of the algebraic integer
matches the degree of the polynomial.
We call an algebraic integer of degree $2$ a {\bf quadratic algebraic
integer}.
Each quadratic algebraic integer has a minimal polynomial of the form
$x^2+bx+c$ with $(b, c) \in {\mathbb Z}^2$.

We now define two sets, ${S}$ and $\bar{S}$, as well as a map $\pi$
from $\bar{S}$ to ${S}$.
We denote by ${S}$ the set of all quadratic algebraic integers in the
open unit interval $(0,1)$.
We denote by $\bar{S}$ the set of all $(b, c) \in {\mathbb Z}^2$
satisfying either $c>0$ and $1+b+c<0$, or $c<0$ and $1+b+c>0$.
If $(b, c) \in \bar{S}$, then $f(x) := x^2 + bx + c$ has exactly one
root of multiplicity 1 in the open unit interval $(0, 1)$ since
$\textrm{sgn} ~f(0) \neq \textrm{sgn} ~f(1)$.
Denoting this root by $\alpha$, it is obvious that $\alpha \in {S}$.
We denote by $\pi$ the map from $\bar{S}$ to $S$ that assigns to each
$(b, c) \in \bar{S}$ the unique $\alpha \in {S}$ that is a root of
$x^2+bx+c$.
This $\pi$ is a bijection.

In \cite{SaitoChaos2016}, it is proposed to use the binary sequence
$\left\{ b_i \right\}_{i=1,2,\dots}$ ($b_i \in \left\{0,\,1\right\}$)
obtained from the binary expansion $\alpha = \sum_{i=1}^{\infty} b_i
2^{-i}$ of $\alpha \in {S}$ as a pseudorandom binary sequence.
Since $\alpha \in {S}$ is irrational, it ensures that $\left\{ b_i
\right\}_{i=1,2,\dots}$ is a nonperiodic sequence.
As mentioned in the introduction, besides the nonperiodicity, there
are several mathematical supports for the high statistical quality of
$\left\{ b_i \right\}_{i=1,2,\dots}$ as a (pseudo-) random binary
sequence.

For any integer $b$ satisfying either $b \ge 1$ or $b \le -3$, we let
\begin{align}\label{eq:SeedSet}
\begin{split}
\bar{I}_b &:= \left\{
\begin{array}{ll}
\left\{(b, c) \in {\mathbb Z}^2 \mid  -b \le c \le -1\right\} &   \textrm{~  if~} b \ge 1, \\
\left\{(b, c) \in {\mathbb Z}^2 \mid  1 \le c \le -b-2\right\} &   \textrm{~  if~} b \le -3,
\end{array}\right.\\
{I}_b &:= \pi \left(\bar{I}_b\right) := \left\{ \pi(b, c) \mid (b, c) \in \bar{I}_b \right\}.
\end{split}
\end{align}
Note that $\bar{I}_b \subset \bar{S}$ and ${I}_b \subset {S}$.
The set ${I}_b$ (equivalently $\bar{I}_b$) is introduced as a set of
seeds for the true orbit generator of \cite{SaitoChaos2016}.
This ${I}_b$ has two properties desirable for a set of seeds, as
described in the following results.

\begin{proposition}[\cite{SaitoChaos2016,SaitoPreprint}]\label{Prop:QuadraticUniformity}
The elements of $I_{b}$
are distributed almost uniformly in the unit interval for sufficiently
large $|b|$.
\end{proposition}

\begin{proposition}[\cite{SaitoPreprint}]\label{Thm:QuadraticNoncoincidence}
Let $b$ be an integer other than $0$, $-1$, and $-2$. Then, ${\mathbb
Q}(\alpha) \neq {\mathbb Q}(\beta)$ holds for all $\alpha, \beta
\in I_{b}$ with $\alpha \neq \beta$.
\end{proposition}

In the context of pseudorandom number generation, the property of
Proposition~\ref{Prop:QuadraticUniformity} is desirable for unbiased
sampling of seeds.
Also, by the property of Proposition~\ref{Thm:QuadraticNoncoincidence}, it is
guaranteed that the binary sequences derived from $I_{b}$ are highly
distinct from each other.
In fact, no identical sequences emerge even after applying to each
binary sequence any operation expressible as a rational map with
rational coefficients (except those mapping elements of $I_{b}$ to
rationals) (cf. \cite[Section 3]{SaitoChaos2016}).

\section{Newton's method}\label{sec:Newton'sMethod}

In this study, we aim to utilize Newton's method to generate binary
sequences identical to those generated by the true orbit generator of \cite{SaitoChaos2016},
despite the fact that Newton's method can only provide approximate
roots of a function.
The formula for the Newton method is given by
\begin{equation}\label{eq:NewtonMethod}
x_{i+1}=F(x_{i}):=x_{i}-\frac{f(x_{i})}{f'(x_{i})} \qquad (i=0,1,\dots),
\end{equation}
where $f(x):=x^2+bx+c$ with $(b, c) \in
\bar{S}$ (cf. the definition of $\bar{S}$ in Section \ref{sec:Preliminaries}).

For simplicity, we focus on the case where $x_{0}=1$ and
$b \ge 1$.
It is easy to see that $f'(x)>0$ and $f''(x)>0$ for $x \in
[\alpha,1]$, where $\alpha$ is a unique root of $f$ satisfying
$0<\alpha<1$.
Thus, $\left\{ x_i \right\}_{i=0,1,2,\dots}$ is a strictly monotone
decreasing sequence converging to $\alpha$.

Let $\epsilon_{i} = x_{i} - \alpha$ ($i=0,1,\dots$).
Note that $0<\epsilon_{i}<1$.
From \eqref{eq:NewtonMethod}, we see that
\[
\epsilon_{i+1}=\epsilon_{i}-\frac{f(\alpha+\epsilon_{i})}{f'(\alpha+\epsilon_{i})} \qquad (i=0,1,\dots).
\]
We have
\[
\epsilon_{i+1}=\frac{\epsilon_{i}^2}{2\alpha + b+2\epsilon_{i}} < \frac{\epsilon_{i}^2}{b} \qquad (i=0,1,\dots),
\]
which implies
\begin{align*}
  \log \epsilon_{i} &< 2^{i} \left(\log \epsilon_{0} - \log b \right) + \log b \qquad (i=1,2,\dots)\\
  &< 2^{i} \left(-\log b \right) + \log b.
\end{align*}
Thus, $\epsilon_{i} < b^{-2^{i}+1}$ holds irrespectively of $c$, and
this inequality also holds for $i=0$.
One can approximate each element of $I_b$
with $b > 1$ with an error less than $2^{-n}$ with $n \ge 0$, if the
number $i$ of iterations of the Newton method satisfies $b^{-2^{i}+1}
\le 2^{-n}$ (cf. definition \eqref{eq:SeedSet}).
That is, for such approximation, it suffices to choose $i$
satisfying
\begin{equation}\label{eq:QuadraticNewtonIterations}
i \ge \log_2 \left( \frac{n}{\log_2 b} + 1\right)
\end{equation}
irrespectively of $c$.

\section{Conditions for binary expansions to match or not}

We now establish conditions under which the first part of the binary
expansion of an irrational number matches that of its upper rational
approximation.  The results obtained in this section will later be
used to ensure that a binary sequence obtained from Newton's method is
identical to a binary sequence obtained from a true orbit generator.

It is known that a real number $x$ has a unique binary expansion
unless $x$ is a dyadic rational number other than 0.
We say that a rational number $x$ is dyadic if $x$ is of the form
$x=m/2^n$ for some integers $m, n$ with $n \ge 0$ (see, e.g.,
\cite{Ko}).
Any dyadic rational other than 0 has precisely two binary expansions,
one ending with all 0s and the other ending with all 1s.

Let $\alpha$ be an irrational number in the open unit interval
$\left(0,1\right)$, and let $x$ be a rational number in the open
interval $\left(\alpha,1\right)$.
Furthermore, assume that there exists an integer $n \ge 2$ such that
$x - \alpha < 2^{-n}$ holds.
In Proposition~\ref{Prop:UnmatchConditions} below, we give necessary
and sufficient conditions under which the first $N$ bits of the binary
expansions of $\alpha$ and $x$ match (or do not match), where $N$ is
an integer satisfying $1 \le N < n$.
We denote the unique binary expansion of $\alpha$ by
\begin{align*}
  \alpha &=\sum_{k=1}^{\infty} b_{k} 2^{-k} \quad \left(b_{k}=b_{k}(\alpha) \in \left\{0, 1\right\}\right)\\
  &= .b_{1}b_{2}\dots.
\end{align*}
Since $x$ is
rational, it may have two binary
expansions.
We denote by
\begin{align*}
  x = .\tilde{b}_{1}\tilde{b}_{2}\dots \quad \left(\tilde{b}_{k}=\tilde{b}_{k}(x) \in \left\{0, 1\right\}\right)
\end{align*}
the binary expansion of $x$ that does not end with all 1s.
Note that non-dyadic $x$ has a unique binary expansion, which does
not end with all 1s.
We define:
\begin{equation}\label{eq:definition_alpha_epsilon}
\alpha_{[1,n]} :=\sum_{k=1}^{n} b_{k} 2^{-k}, \qquad \alpha_{[n+1,\infty]} :=\sum_{k=n+1}^{\infty} b_{k} 2^{-k}, \qquad \epsilon := x - \alpha.
\end{equation}
Since $0<\alpha_{[n+1,\infty]}< 2^{-n}$, we have
\begin{equation}\label{eq:alpha_epsilon}
0<\alpha_{[n+1,\infty]} + \epsilon< 2^{-(n-1)}.
\end{equation}
Note that $\alpha_{[n+1,\infty]} + \epsilon$ $(= x -
\alpha_{[1,n]})$ is rational.
Let $.\tilde{c}_{1}\tilde{c}_{2}\dots$ be the binary expansion of
$\alpha_{[n+1,\infty]} + \epsilon$ that does not end with all 1s.

\begin{proposition}\label{Prop:UnmatchConditions}
Let $\alpha \in \left(0,1\right)$ be irrational with the
binary expansion $.b_{1}b_{2}\dots$.
Let $x \in \left(\alpha,1\right)$ be rational, and let
$.\tilde{b}_{1}\tilde{b}_{2}\dots$ be the binary expansion of $x$ that
does not end with all 1s.
Let $N$ and $n$ be integers satisfying $1 \le N < n$ and
$x - \alpha < 2^{-n}$.
Let $\alpha_{[n+1,\infty]}$ and $\epsilon$ be given by \eqref{eq:definition_alpha_epsilon}, and let
$.\tilde{c}_{1}\tilde{c}_{2}\dots$ be the binary expansion of
$\alpha_{[n+1,\infty]} + \epsilon$ that does not end with all 1s.
The following conditions are equivalent:
\begin{enumerate}[(i)]
\item $b_{1}b_{2} \dots b_{N} \neq \tilde{b}_{1}\tilde{b}_{2}\dots
  \tilde{b}_{N}$;
\item $b_{N+1}=b_{N+2}=\dots=b_{n}=1$ and $\tilde{c}_{n}=1$;
\item $\tilde{b}_{N+1}=\tilde{b}_{N+2}=\dots=\tilde{b}_{n}=0$ and $\tilde{c}_{n}=1$.
\end{enumerate}
\end{proposition}

\begin{proof}
The expansions $.\tilde{b}_{1}\tilde{b}_{2}\dots$ and
$.\tilde{c}_{1}\tilde{c}_{2}\dots$ do not end with all 1s, and since
\[
\sum_{k=1}^{\infty} \tilde{b}_{k} 2^{-k} = x = \alpha_{[1,n]} + \alpha_{[n+1,\infty]} + \epsilon = \sum_{k=1}^{n} b_{k} 2^{-k} + \sum_{k=1}^{\infty} \tilde{c}_{k} 2^{-k},
\]
we have $\tilde{b}_{k} = \tilde{c}_{k}$ for all $k \ge n+1$.
By \eqref{eq:alpha_epsilon}, we also have $\tilde{c}_{k} = 0$ for all
$k$ satisfying $1 \le k \le n-1$.
As a result, we have
\begin{equation}\label{eq:relation_tildeb_b_tildec}
\sum_{k=1}^{n} \tilde{b}_{k} 2^{-k} = \sum_{k=1}^{n} b_{k} 2^{-k} + \tilde{c}_{n} 2^{-n}.
\end{equation}
Thus, $b_{1}b_{2} \dots b_{N} \neq \tilde{b}_{1}\tilde{b}_{2}\dots
\tilde{b}_{N}$ if and only if
\[
\sum_{k=N+1}^{n} b_{k} 2^{-k} + \tilde{c}_{n} 2^{-n} \ge 2^{-N}.
\]
This proves the equivalence of (i) and (ii).

It is obvious that (ii) implies (iii).
Suppose that (iii) holds.
Then, \eqref{eq:relation_tildeb_b_tildec} gives
\[
\sum_{k=N+1}^{n} b_{k} 2^{-k} + 2^{-n} \equiv 0 \pmod{2^{-N}},
\]
which implies
\[
\sum_{k=N+1}^{n} b_{k} 2^{-k} \equiv  2^{-N} - 2^{-n} \pmod{2^{-N}}.
\]
Since
\[
0 \le \sum_{k=N+1}^{n} b_{k} 2^{-k} \le \sum_{k=N+1}^{n} 2^{-k} =  2^{-N} - 2^{-n},
\]
we obtain (ii).
\end{proof}

We remark that the proposition still holds if the condition on $\alpha
\in \left(0,1\right)$ is relaxed from irrational to not dyadic
rational.

Proposition~\ref{Prop:UnmatchConditions} has the following
corollary, which can be used to ensure that a binary sequence obtained from
Newton's method is identical to a binary sequence obtained from a true
orbit generator.
Let $n$ be an integer satisfying $n \ge 2$ and $x - \alpha < 2^{-n}$.
Now, suppose that there exists $k$ with $2 \le k \le n$ such that
$\tilde{b}_{k}=1$, and let $N = k-1$.
Then, Condition (iii) of Proposition~\ref{Prop:UnmatchConditions}
does not hold, which leads us to the following statement:

\begin{corollary}\label{Cor:QuadraticUniformity}
Let $\alpha$ and $x$ be as in Proposition~\ref{Prop:UnmatchConditions}.
Let $n$ be an integer satisfying $n \ge 2$ and $x - \alpha < 2^{-n}$.
If there exists $k$ with $2 \le k \le n$ such that $\tilde{b}_{k}=1$,
then $b_{1}b_{2} \dots b_{k-1} = \tilde{b}_{1}\tilde{b}_{2}\dots
  \tilde{b}_{k-1}$ holds.
\end{corollary}

\section{Algorithm}\label{sec:Algorithm}

We now explain our algorithm using Newton's method to generate binary
sequences identical to those generated by the true orbit generator of \cite{SaitoChaos2016}.
Since we only consider polynomials with integer coefficients for $f$
in \eqref{eq:NewtonMethod}, $F$ in the same equation is a rational
function with integer coefficients.
As noted in the previous section, any term of the sequence $\left\{x_i
\right\}_{i=0,1,2,\dots}$, defined by the recurrence relation
\eqref{eq:NewtonMethod} with the initial condition $x_{0}=1$, is
rational, and thus one can exactly evaluate $x_i$ using arbitrary
precision arithmetic on rational numbers or integers.
Unlike computing a true orbit of a linear or linear fractional map,
the cost required to compute a true orbit of $F$ is extremely high
(see \cite{SaitoChaos2015}).
However, in order to obtain a reasonably long pseudorandom binary sequence, it
is sufficient to iterate $F$ several dozen times
(for example, the number of iterations of $F$ required to generate a
pseudorandom sequence of length $2^{36} - 2$, which is discussed in
Subsection \ref{subsec:StatisticalTesting}, is 36).

Below, we will explain the algorithm using
arbitrary precision rational arithmetic
(see Section \ref{sec:ComputationalComplexity} for the use of arbitrary
precision integer arithmetic).
Note that $F$ is given by
\[
F(x)=\frac{x^2 - c}{2 x +b}.
\]

For the pseudorandom number generation, we first select a seed $(b, c)
\in \bar{S}$ with $b > 1$ and two positive integers $N$ and $n'$,
where $N$ is the length of a binary sequence to be generated and $n'$
is used as a margin to avoid additional iterations of $F$ in the later
stage (to be explained below).
The procedure for selecting a seed is called initialization (or
randomization) \cite[Section 2.4]{Sugita}, and one such procedure
is the following:
If one can choose any element of a set consisting of $k$ elements with
equal probability, then select one element from $I_{k}$ (or equivalently $\bar{I}_k$) with
equal probability and use it as a seed (see definition \eqref{eq:SeedSet}).
Note that $\left|I_{k}\right|=\left|\bar{I}_{k}\right|=k$ for $k>1$.

After the initialization, we compute the first terms of $\left\{x_i
\right\}_{i=0,1,2,\dots}$, defined by \eqref{eq:NewtonMethod} with
$x_{0}=1$.
Specifically, we put
\begin{align}\label{eq:Istar}
  n := N + n',\qquad
  i^{*} := \left\lceil \log_2 \left( \frac{n}{\log_2 b} + 1\right) \right\rceil,
\end{align}
with reference to \eqref{eq:QuadraticNewtonIterations}.
Here $\left\lceil x \right\rceil$ denotes the smallest integer which
is not less than a real number $x$.
Then, we exactly compute $x_{i^{*}} = F^{\left(i^{*}\right)}(1)$,
where $F^{\left(i\right)}$ denotes the $i$-fold composition of $F$.

Then, we compute a binary sequence $\tilde{b}_{1}\tilde{b}_{2} \dots
\tilde{b}_{n^{*}}$, namely the first $n^{*}$ bits of the binary
expansion of $x_{i^{*}}$, where we define an integer $n^{*}$ by
\[
n^{*} := \left\lfloor \log_2 b^{2^{i^{*}} -1} \right\rfloor,
\]
where we put $\left\lfloor x \right\rfloor := -\left\lceil -x
\right\rceil$.
Obviously, $n \le n^{*}$ holds.

Lastly, we attempt to output a pseudorandom sequence of length $N$,
identical to the one generated by the true orbit generator.
That is, we try to exactly extract $b_{1}b_{2} \dots b_{N}$, namely
the first $N$ bits of the binary expansion of the quadratic algebraic
integer $\alpha := \pi(b, c)$, from
$\tilde{b}_{1}\tilde{b}_{2} \dots \tilde{b}_{n^{*}}$ currently in
hand.
We define an integer $k^{*}$ by
\[
k^{*} := \max \left\{ k \in {\mathbb Z} ~\middle|~ k \le n^{*}, \tilde{b}_{k}=1 \right\}.
\]
If $N \le k^{*}-1$, then we output $\tilde{b}_{1}\tilde{b}_{2} \dots
\tilde{b}_{N}$, i.e., the first $N$ bits of
$\tilde{b}_{1}\tilde{b}_{2} \dots \tilde{b}_{n^{*}}$, as a
pseudorandom sequence.
Note that by Corollary~\ref{Cor:QuadraticUniformity},
$\tilde{b}_{1}\tilde{b}_{2} \dots \tilde{b}_{N}$ is guaranteed to be
identical to the $N$-bit binary sequence $b_{1}b_{2} \dots b_{N}$ obtained from the true orbit
generator with the seed $(b, c)$.
Otherwise, i.e., if $N \ge k^{*}$, we increment the value of
$i^{*}$ by $1$, and retry the exact computation of $x_{i^{*}}$.

Here we summarize the algorithm, implemented using the variable $x$
holding an arbitrary-precision rational number.

\begin{algorithm}[H]
    \caption{Generate pseudorandom binary sequence}
    \label{AlgorithmRational}
    \begin{algorithmic}[1]
    \REQUIRE two positive integers $N$, $n'$, and a seed $(b, c)$ with $b > 1$
    \ENSURE an $N$-bit binary sequence $b_{1}b_{2} \dots b_{N}$
    \STATE $n := N + n'$
    \STATE $i^{*} := \left\lceil \log_2 \left( \displaystyle\frac{n}{\log_2 b} + 1\right) \right\rceil$
    \STATE $x := 1$

    \FOR{$i:=0$ \TO $i^{*}-1$}
    \STATE $x :=\displaystyle\frac{x^2 - c}{2 x +b}$ \label{LineOfApplyingF}
    \STATE $i := i + 1$
    \ENDFOR

    \STATE $n^{*} := \left\lfloor \log_2 b^{2^{i} -1} \right\rfloor$
    \STATE $\textrm{\bf digits} := \textrm{AllocateMemoryOfSize}(n^{*})$
    \STATE Store the first $n^{*}$ bits of the binary expansion of $x$ in $\textrm{\bf digits}$
    \STATE $k^{*} := \textrm{LastIndexOf}(\textrm{\bf digits},1)$
    
    \IF{$N \le k^{*}-1$}
    \STATE Output the first $N$ bits of $\textrm{\bf digits}$
    \ELSE
    \STATE \textbf{goto} \ref{LineOfApplyingF} \label{Goto}
    \ENDIF
    \end{algorithmic}
\end{algorithm}

When
$\tilde{c}_{n^{*}}=0$ and $b_{N+1}=b_{N+2}=\dots=b_{n^{*}}=0$, or when
$\tilde{c}_{n^{*}}=1$ and $b_{N+1}=b_{N+2}=\dots=b_{n^{*}}=1$,
there is no $k$ in the
range $N+1 \le k \le n^{*}$ such that $\tilde{b}_{k}=1$, and
the goto-statement \ref{Goto} is executed.
We regard the bits $b_{1}, b_{2}, \dots$ in the binary expansion of
$\alpha$ as (a typical sample of) independently identically
distributed random variables.
Then, the probability of there being no $k$ in the range $N+1 \le k
\le n^{*}$ such that $\tilde{b}_{k}=1$ is independent of the value of
$\tilde{c}_{n^{*}}$ and is given by $2^{-(n^{*}-N)}$.
Therefore, by repeated execution of the goto-statement, the
probability of the goto-statement being newly executed becomes
arbitrarily small.
However, it is also possible to make the probability of the
goto-statement being executed arbitrarily small by taking $n'$ to be
large beforehand (note that $2^{-(n^{*}-N)} \le 2^{-n'}$).
Below, following the latter approach, we assume that one inputs a
sufficiently large positive integer as $n'$.
Accordingly, a part corresponding to lines 14 and 15 of Algorithm
\ref{AlgorithmRational} will be excluded from the algorithm considered
for computational complexity in Section
\ref{sec:ComputationalComplexity} and the program used for numerical
experiments in Section \ref{sec:Experiments}.

\section{Computational complexity}\label{sec:ComputationalComplexity}

In this section we explore the time complexity of the algorithm
discussed in the previous section, which uses Newton's method to
generate pseudorandom binary sequences.
We denote by $p_i$ and $q_i$ the numerator and denominator,
respectively, of each term $x_i$ in the rational sequence $\left\{x_i
\right\}_{i=0,1,2,\dots}$ defined by \eqref{eq:NewtonMethod} with
$x_{0}=1$.
If we do not consider reduction, the sequence
$\left\{\left(p_i,q_i\right) \right\}_{i=0,1,2,\dots}$ is given by
\begin{equation}\label{eq:RecurrenceRelationForPandQ}
\left\{
\begin{aligned}
&p_{0}=1,& &p_{i+1}=p_{i}^2 - c q_{i}^2 & &(i \ge 0),\\
&q_{0}=1,& &q_{i+1}=2 p_{i} q_{i} + b q_{i}^2 & &(i \ge 0).
\end{aligned}
\right.
\end{equation}
The algorithm described in Section \ref{sec:Algorithm}, when
implemented using two variables $p$ and $q$ that hold
arbitrary-precision integers, becomes as follows.

\begin{algorithm}[H]
    \caption{Generate pseudorandom binary sequence using arbitrary-precision integers (goto-free)}
    \label{AlgorithmIntegers}
    \begin{algorithmic}[1]
    \REQUIRE two positive integers $N$, $n'$, and a seed $(b, c)$ with $b > 1$
    \ENSURE an $N$-bit binary sequence $b_{1}b_{2} \dots b_{N}$
    \STATE $n := N + n'$
    \STATE $i^{*} := \left\lceil \log_2 \left( \displaystyle\frac{n}{\log_2 b} + 1\right) \right\rceil$
    \STATE $\left(p,q\right) := \left(1,1\right)$

    \FOR{$i:=0$ \TO $i^{*}-1$}  \label{LineStartLoop}
    \STATE $\left(p,q\right) := \left(p^2 - c q^2,2 p q + b q^2\right)$  \label{LineOfApplyingFforPQ}
    \STATE $i := i + 1$
    \ENDFOR \label{LineEndLoop}

    \STATE $n^{*} := \left\lfloor \log_2 b^{2^{i} -1} \right\rfloor$
    \STATE $\textrm{\bf digits} := \textrm{AllocateMemoryOfSize}(n^{*})$
    \STATE Store the first $n^{*}$ bits of the binary expansion of $p/q$ in $\textrm{\bf digits}$ \label{LineDivision}
    \STATE $k^{*} := \textrm{LastIndexOf}(\textrm{\bf digits},1)$
    
    \IF{$N \le k^{*}-1$}
    \STATE Output the first $N$ bits of $\textrm{\bf digits}$
    \ENDIF
    \end{algorithmic}
\end{algorithm}

Lines 3, 5, and 10 of Algorithm \ref{AlgorithmRational} are modified,
and lines 14 and 15 are deleted (see discussion at the end of Section
\ref{sec:Algorithm}).

In order to analyze the time complexity of the algorithm, we first
evaluate the lengths of the binary expansions of $p_i$ and $q_i$ in
\eqref{eq:RecurrenceRelationForPandQ}.
If $p_i$ and $q_i$ are both positive, then $p_{i+1}$ and $q_{i+1}$ are
also positive since $b > 0$ and $c < 0$.
Since $p_0$ and $q_0$ are both positive, $p_i$ and $q_i$ are positive
for all $i = 0,1,2,\dots$.
By \eqref{eq:RecurrenceRelationForPandQ} and $1 \le p_i \le q_i$, we
have
\begin{equation*}
b q_i^2 < q_{i+1} \le (b+2) q_i^2,
\end{equation*}
which, together with the initial condition, gives
\begin{equation}\label{eq:q_i}
b^{2^{i} -1} \le q_{i} \le (b+2)^{2^{i} -1},
\end{equation}
for every $i = 0,1,2,\dots$.
Since $q_i$ is a positive integer, the length of its binary expansion,
denoted by $\ell_i$, is given by
\[
\ell_i =
\left\lfloor \log_{2} q_i \right\rfloor +1.
\]
By \eqref{eq:q_i}, we have
\begin{equation}\label{eq:l_i}
  \left(2^{i} -1\right) \log_{2} b  < \ell_i \le \left(2^{i} -1\right) \log_{2} \left(b+2 \right) + 1 < 2^{i} \log_{2} \left(b+2 \right)
\end{equation}
for all $i = 0,1,2,\dots$.
We also see from \eqref{eq:RecurrenceRelationForPandQ} and $b \ge 2$
that $q_{i+1} \ge 4 q_i$ holds for all $i$, and thus the sequence
$\left\{ \ell_i \right\}_{i=0,1,2,\dots}$ is strictly monotonically
increasing.

We now analyze the time complexity of the algorithm, that is, we
establish an upper bound on the worst-case time required by the
algorithm.
We assume that bit operations, such as multiplication of two bits,
take a constant time.
Suppose that the time complexity of the multiplication of two
$\ell$-bit integers is expressed as $O\left(\ell g(\ell)\right)$,
where $O$ is Landau's Big O notation, and $g$ is a real-valued
function defined on the positive integers, with values always greater
than or equal to $1$.
We further assume that for any $c \ge 1$, there exists a constant
$c' \ge 1$ such that $g(c \ell) \le c'g(\ell)$ holds for all
sufficiently large $\ell$.
Various algorithms have been proposed for multiplying large numbers.
For example, the Karatsuba-Ofman algorithm \cite{Karatsuba} has time
complexity $O(\ell^{\log_2 3})$, the Toom-Cook or Toom-3 algorithm
\cite{Toom} has time complexity $O(\ell^{\log_3 5})$, and the
Sch\"{o}nhage-Strassen algorithm \cite{Schoenhage} has time complexity
$O(\ell \cdot \log\ell \cdot \log\log\ell)$.
The time complexities of these are expressed in the form of $O\left(\ell g(\ell)\right)$ above.
Note that if we follow the convention that $\log\ell$ is interpreted
as $\max \left\{\log\ell, 1\right\}$, then the term $\log\ell \cdot
\log\log\ell$ within the time complexity of the Sch\"{o}nhage-Strassen
algorithm always take values greater than or equal to $1$.
Below $c_1, c_2, \dots$ represent certain positive constants that do
not require further specification.
The time complexities for the addition and subtraction of two
$\ell$-bit integers, as well as the multiplication of an $\ell$-bit
integer with a constant integer, are $O(\ell)$.
Therefore, at line \ref{LineOfApplyingFforPQ} of Algorithm
\ref{AlgorithmIntegers}, the multiplications of two large integers
take the largest running time.
Since $g(\ell) \ge 1$ for all $\ell \ge 1$, the time spent in the loop
of lines \ref{LineStartLoop}-\ref{LineEndLoop} is bounded above by
\begin{equation}\label{eq:TimeforLoop}
c_1 \sum_{i=0}^{i^{*}-1} \ell_i g\left(\ell_i\right).
\end{equation}
Since $g(\ell_i) \le g(\ell_{i^{*}-1})$ and $\ell_i \le c_2 2^i$ (see
\eqref{eq:l_i}) for all $i=0,1,\dots, i^{*}-1$, we see that
\begin{equation*}
\sum_{i=0}^{i^{*}-1} \ell_i g\left(\ell_i\right) \le g\left(\ell_{i^{*}-1}\right)\sum_{i=0}^{i^{*}-1} \ell_i \le c_2 g\left(\ell_{i^{*}-1}\right) \left(2^{i^{*}}-1\right) \le c_3 g\left(\ell_{i^{*}-1}\right) 2^{i^{*}-1}.
\end{equation*}
By \eqref{eq:l_i}, we have $2^{i^{*}-1} <
\left(\ell_{i^{*}-1}/\log_{2} b\right) +1$, which yields
\begin{equation}\label{eq:Sumlgl}
\sum_{i=0}^{i^{*}-1} \ell_i g\left(\ell_i\right) \le c_4 \ell_{i^{*}-1} g\left(\ell_{i^{*}-1}\right) + c_3 g\left(\ell_{i^{*}-1}\right) \le c_5 \ell_{i^{*}-1} g\left(\ell_{i^{*}-1}\right).
\end{equation}
By \eqref{eq:Istar}, we have
$i^{*} -1 < \log_2 \left( n\left(\log_2 b\right)^{-1} + 1\right)$,
which combined with \eqref{eq:l_i} gives
\begin{equation*}
  \ell_{i^{*}-1} \le \left(2^{i^{*}-1} -1\right) \log_{2} \left(b+2 \right) + 1
  < \frac{\log_{2} \left(b+2 \right)}{\log_2 b} n +1 \le 2 n +1.
\end{equation*}
Thus, $\ell_{i^{*}-1} \le 2 n = 2 N + 2 n'$.
Assuming that $n'$ is kept constant (see discussion at the end of
Section \ref{sec:Algorithm}), we have
\begin{equation}\label{eq:listar-1}
\ell_{i^{*}-1} \le c_6 N
\end{equation}
for all sufficiently large $N$.
By \eqref{eq:Sumlgl}, \eqref{eq:listar-1}, and the assumption on $g$,
we see that the time spent in the loop of lines
\ref{LineStartLoop}-\ref{LineEndLoop}, bounded by
\eqref{eq:TimeforLoop}, is $O\left(N g(N)\right)$.
It is straightforward to see that the time spent at line
\ref{LineDivision} of Algorithm \ref{AlgorithmIntegers} is also
$O\left(N g(N)\right)$.
In fact, we can show, similarly to \eqref{eq:listar-1}, that 
\begin{equation*}
\ell_{i^{*}} \le c_7 N, \qquad n^{*} \le c_8 N
\end{equation*}
for all sufficiently large $N$.
It is known that multiple-precision division is linearly equivalent to
multiple-precision multiplication \cite{Brent}.
Thus, the generation of the first $n^{*}$ bits of the binary expansion
of $x_{i^{*}}=p_{i^{*}}/q_{i^{*}}$ takes $O\left(N g(N)\right)$ time
complexity.
Therefore, the time complexity of our algorithm to produce a
pseudorandom sequence of length $N$ is $O\left(N g(N)\right)$, the
same as that of the multiplication of two $N$-bit integers.
We remark that the algorithm is more efficient, at least for large
$N$, compared to the true orbit generator presented in
\cite{SaitoChaos2016} which has time complexity
$O(N^2)$ (see Appendix \ref{sec:AppendixComplexityTrueOrbitGenerators}).

\section{Experiments}\label{sec:Experiments}

In this section, we report the results of numerical experiments on
the proposed method, focusing on the computation time and memory usage
required for pseudorandom number generation, as well as the
statistical testing of a generated pseudorandom sequence.
The program used for the numerical experiments is written in {\tt C},
and it utilizes the GNU MP and MPFR libraries for arbitrary precision
arithmetic.
Table \ref{tab:system_specifications} shows the specifications of the
computer and the software used in the numerical experiments.

\begin{table}
\caption{System specifications}
\label{tab:system_specifications}
\centering
\begin{tabular}{|l|l|}
\hline
\textbf{Item}                 & \textbf{Specification} \\ \hline
CPU                           & Intel(R) Core(TM) i5-14600K \\ \hline
MEMORY                        & 192GB (DDR5 4200MHz)          \\ \hline
HDD/SSD                       & 1TB (NVMe M.2)             \\ \hline
OS                            & Ubuntu 22.04.4 LTS         \\ \hline
C Compiler                    & gcc 11.4.0                 \\ \hline
Arithmetic Library            & GMP 6.2.1, GNU MPFR 4.1.0  \\ \hline
Randomness Testing Utility    & TestU01 1.2.3              \\ \hline
\end{tabular}
\end{table}

\subsection{Computation time}

Figure \ref{fg:ComputationTime} shows the results of the computation time for pseudorandom
number generation.
Here, we set $b=2$, $c=-1$, $n'=1$, and $N=\lceil 2^k-1 \rceil$, where
we vary the value of $k$ by 1/4 in the range $7 \leq k < 27$.
Note that when $b=2$ and $n'=1$, the number of iterations of Newton's
method increases by one at $N$ having integer values of $k$
(cf. \eqref{eq:Istar}).
The computation time (elapsed time) is calculated by subtracting the
start time of the program's main function from the time at which the
output of the generated pseudorandom sequence is completed.
In the measurement of computation time, we execute the program 100
times for each $N$ and calculate the average computation time.
The cache is cleared before each run of the program to eliminate
the influence of previous executions.
For comparison, we also measure the computation time of the true orbit
generator of \cite{SaitoChaos2016} for $N$ with $k$ in the range $7
\le k \le 19.5$.
In the proposed method using Newton's method, the number of iterations
of Newton's method increases according to \eqref{eq:Istar} as $N$
increases, causing the computation time to increase in a stepwise
manner, as seen in Figure \ref{fg:ComputationTime}.

\begin{figure}
\begin{center}
\includegraphics[width=0.618\textwidth,clip]{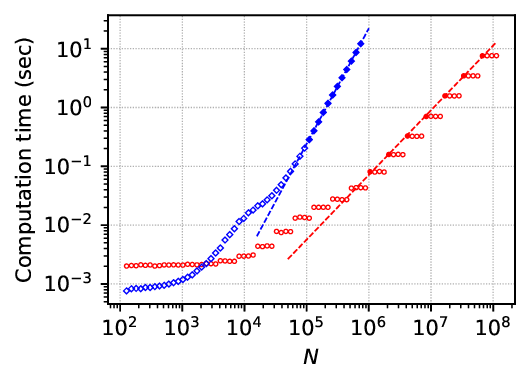}
\end{center}
\caption{\label{fg:ComputationTime}Results of computation times for
the proposed method using Newton's method (red circles) and the
quadratic true orbit generator (blue diamonds)}
\end{figure}

To evaluate the growth rate of computation time as $N$ increases, we
conducted a power approximation for the computation time of the
proposed method for $k = 20, 21, \dots, 26$ (represented by filled
markers in Figure \ref{fg:ComputationTime}) .
For the true orbit generator, we conducted a power approximation in
the range $16.75 \leq k \leq 19.5$.
Consequently, within the considered range of $N$, the computation time
for the proposed method grows like $N^{1.10}$, up to constant factors.
This is consistent with the fact that in GMP, the multiplication of
large integers is performed using an algorithm based on the
Sch\"{o}nhage-Strassen algorithm, and that the time complexities of
these algorithms are strongly sub-quadratic but super-linear
(cf. Section \ref{sec:ComputationalComplexity}).
On the other hand, the growth rate of computation time when using the
true orbit generator is approximately $N^{1.96}$,
similar to the time complexity $O(N^2)$ of the true orbit generator
(see Appendix \ref{sec:AppendixComplexityTrueOrbitGenerators}).
In any case, it is confirmed that the proposed method using Newton's
method significantly speeds up the computation time.

\subsection{Memory usage}

In Figure \ref{fg:MemoryUsage}, we show the results of the memory usage for
pseudorandom number generation.
The parameters are the same as in the previous subsection.
For measuring the memory usage, the Maximum Resident Set Size (MaxRSS)
was obtained using the {\tt getrusage} function at the point when a
pseudorandom sequence was output.
MaxRSS represents the maximum amount of memory used by a job on
physical memory (RAM).
Even for $k = 26.75$, the MaxRSS was approximately 386 MB, which is
sufficiently smaller than 192 GB available on the system.
Therefore, MaxRSS can be considered as the amount of memory used by
  the program during execution.
For the true orbit generator, data in the range $19.75 \leq k \leq 23.5$
is also included, but due to the longer computation time, the value of
one sample is used instead of the average of 100 samples (the green
diamonds in Figure \ref{fg:MemoryUsage}).

\begin{figure}
\begin{center}
\includegraphics[width=0.618\textwidth,clip]{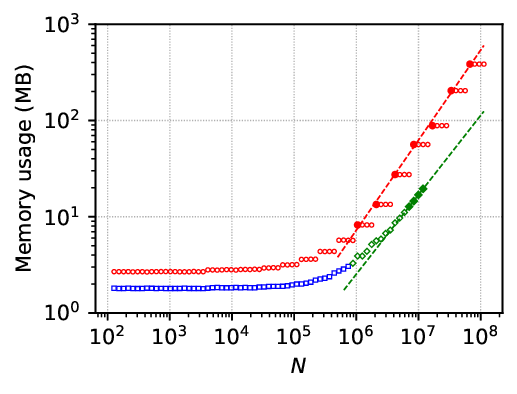}
\end{center}
\caption{\label{fg:MemoryUsage}Results of memory usages for
the proposed method using Newton's method (red circles) and the
quadratic true orbit generator (blue squares and green diamonds).
The red circles and blue squares represent the averages of 100 samples, whereas the
green diamonds represent the values obtained from a single sample.}
\end{figure}

As in the previous subsection,
a power approximation was
performed for the proposed method with $k = 20, 21, \dots, 26$, and
for the true orbit generator with $k = 22.75, 23, 23.25, 23.5$.
As a result, the growth rate of memory usage with respect to the
sequence length $N$ is approximately $N^{0.94}$ for the proposed
method and approximately $N^{0.82}$ for the true orbit generator.
We see from \eqref{eq:l_i} and \eqref{eq:b_n} that the space
complexity (memory usage) for generating a pseudorandom sequence of
length $N$ is $O(N)$ for both the proposed method and the true orbit
generator.
Note that when $b = 2$ and $n' = 1$, the number $i^{*}$ of iterations
of the Newton method for generating a pseudorandom sequence of length
$N = 2^k - 1$ ($k \in {\mathbb Z}$) using the proposed method is given
by $k + 1 = \log_2 (N + 1) + 1$ (cf. \eqref{eq:Istar}).
The memory usage estimated from the experiment is less than $O(N)$,
which is considered to be due to the influence of the memory used by
linked libraries.
As $N$ increases, this contribution relatively decreases, and it is
expected to approach $O(N)$.

\subsection{Statistical testing}\label{subsec:StatisticalTesting}

Finally, we report the result of evaluating the statistical
properties of a generated pseudorandom sequence.
For the evaluation, we generated the pseudorandom sequence of length
$N = 2^{36} - 2 = 68,719,476,734$ with $b = 2$, $c = -1$, and $n' =
1$ using the proposed method.
It should be emphasized that generating a pseudorandom sequence of
this length is practically impossible with the true orbit generator
due to the excessive computation time required.
Randomness tests were conducted using TestU01's RabbitFile \cite{L'Ecuyer2},
which consists of 26 randomness tests.
As a result, for all tests in RabbitFile, no suspicious
$p$-values---such as extremely large or small ones---were observed,
and the sequence successfully passed all 26 tests, demonstrating good
statistical properties.

\section*{Acknowledgements}

We thank Saul Schleimer for introducing us to relevant literature.
This research was supported by JSPS KAKENHI Grant Number JP22K12197.

\appendix

\section{Quadratic true orbit generator and its computational complexity}\label{sec:AppendixComplexityTrueOrbitGenerators}

In this appendix, we briefly explain the pseudorandom number generator
utilizing chaotic true orbits of the Bernoulli map on quadratic
algebraic integers proposed in \cite{SaitoChaos2016}, and then discuss
the time complexity of this generator.

The transformation $\bar{M}_B := \pi^{-1} \circ M_B \circ \pi$ on
$\bar{S}$ corresponding to the Bernoulli map $M_B(x):= 2x \pmod{1}$ on
${S}$ is given by (see Section \ref{sec:Preliminaries} for the
definitions of ${S}$, $\bar{S}$, and $\pi$):

\begin{subequations}\label{eq:Transformation}
\noindent
if $\textrm{sgn}\left(1+2b+4c\right) \neq \textrm{sgn} ~c$,
\begin{align*}
  \bar{M}_B: \left(
\begin{array}{c}
b\\
c
\end{array}
\right)
&\longmapsto
\left(
\begin{array}{cc}
2	&	0\\
0	&	4
\end{array}
\right)
\left(
\begin{array}{c}
b\\
c
\end{array}
\right);\\
\intertext{otherwise}
\bar{M}_B: \left(
\begin{array}{c}
b\\
c
\end{array}
\right)
&\longmapsto
\left(
\begin{array}{cc}
2	&	0\\
2	&	4
\end{array}
\right)
\left(
\begin{array}{c}
b\\
c
\end{array}
\right)
+
\left(
\begin{array}{c}
2\\
1
\end{array}
\right),
\end{align*}
\end{subequations}
see \cite{SaitoChaos2016}.
The true orbit generator of \cite{SaitoChaos2016} generates a
pseudorandom binary sequence $\left\{ \epsilon_n
\right\}_{n=0,1,\cdots, N-1}$ ($\epsilon_n \in \left\{0,\,1\right\}$)
of length $N$ for a given seed $(b_0, c_0) \in \bar{S}$ as follows:
It exactly computes a sequence of length $N$, $\left\{ (b_n, c_n)
\right\}_{n=0,1,\cdots, N-1}$, satisfying the recurrence relation
\begin{equation}\label{eq:BernoulliMap}
(b_{n+1}, c_{n+1})=\bar{M}_B (b_n, c_n) \qquad (n \ge 0).
\end{equation}
We call this computation {\bf true orbit computation} \cite{SaitoPhysD,SaitoChaos2015}.
Each $\epsilon_n$ ($0 \le n \le N-1$) in the pseudorandom binary
sequence is defined
as follows: if $\textrm{sgn}\left(1+2b_n+4c_n\right) \neq \textrm{sgn}
~c_n$, then $\epsilon_n :=0$; otherwise $\epsilon_n :=1$.

In order to study the time complexity of the generator, we evaluate
the lengths of the binary expansions of $b_n$ and $c_n$ in
\eqref{eq:BernoulliMap}.
It is easy to see that
\begin{equation}\label{eq:b_n}
2^{n} b_0 \le b_n \le 2^{n} (b_0 + 2) -2
\end{equation}
for every $n \ge 0$.
Assume now that $b_0 > 0$, as assumed in the main text.
Then, we see that $b_n > 0$ for all $n$.
Thus, the length of the binary expansion of $b_n$ is given by
$\left\lfloor \log_{2} b_n \right\rfloor +1$, and we see from
\eqref{eq:b_n} that it is $O(n)$.
The length of the binary expansion of $c_n$ is also $O(n)$ since $-b_n
\le c_n \le -1$ when $b_n > 0$.

The true orbit computation \eqref{eq:BernoulliMap} consists of two
basic steps: the evaluation of the sign of $1+2b_n+4c_n$, and the
application of a linear transformation to $(b_n, c_n)$.
Note that the sign of $c_n$ is negative, as seen above.
In addition, $\epsilon_n$ other than $\epsilon_{N-1}$ is obtained as a
byproduct of the sign evaluation.
Both the sign evaluation and the application of
a linear transformation involve two operations: Multiplication of an
arbitrary-precision integer by a constant integer, and addition of two
arbitrary-precision integers.
If the arbitrary-precision integers consist of $\ell$ bits, the
worst-case times to perform these operations are $O(\ell)$, as
described in Section \ref{sec:ComputationalComplexity} (we assume that
bit operations, such as multiplication of two bits, take a constant
time).

In order to generate a pseudorandom sequence of length $N$, the sign
evaluation is performed for every $(b_n, c_n)$ with $n=0,1,\dots,N-1$,
and the application of a linear transformation is performed for every
$(b_n, c_n)$ with $n=0,1,\dots,N-2$.
Since the lengths of the binary expansions of $b_n$ and $c_n$ are
$O(n)$, the (worst-case) time complexity of the true orbit generator
is $O(N^2)$.

\end{document}